\documentclass[11pt]{article}

\usepackage[margin=1in]{geometry}
\usepackage{amsmath,amssymb,amsthm,mathtools}
\usepackage{booktabs}
\usepackage{enumitem}
\usepackage{array}
\usepackage[colorlinks=true,linkcolor=blue,citecolor=blue,urlcolor=blue]{hyperref}
\allowdisplaybreaks
\usepackage{natbib}

\newtheorem{definition}{Definition}
\newtheorem{proposition}{Proposition}

\newtheorem{remark}{Remark}
\newtheorem{corollary}{Corollary}

\newcommand{\E}{\mathbb E}
\newcommand{\given}{\mid}

\newcommand{\inn}{\mathrm{in}}
\newcommand{\CE}{\mathrm{CE}}

\newcommand{\CO}{\mathrm{CO}}

\title{Correlated optimin}
\author{Mehmet Mars Seven\thanks{Department of Political Economy, King's College London, UK. E-mail: mehmet.mars.seven@kcl.ac.uk}}
\date{}

\begin{document}
	\maketitle
	
	\begin{abstract}
		We extend the optimin notion of \citet{ismail2025} from mixed strategy profiles to correlated distributions. A correlated distribution is evaluated by the worst expected payoff each player can receive when opponents may either obey their private recommendations or make unilateral recommendation-contingent deviations that are strictly profitable under the posterior induced by the distribution. Correlated optimins are Pareto optimal with respect to this vector of guaranteed payoffs. We show that a correlated optimin exists in every finite game. In addition, for every correlated equilibrium, there exists a correlated optimin such that every player's guaranteed payoff is weakly higher than his or her correlated equilibrium payoff. In two-player zero-sum games, correlated optimin coincides with correlated equilibrium and yields the maximin value. Outside zero-sum games, correlated optimin may strictly improve upon all correlated equilibria. We illustrate this with a simple $2\times2$ game with a unique correlated and coarse correlated equilibrium, in which there exists a correlated optimin that strictly Pareto dominates the equilibrium payoff. JEL: C70, D81
	\end{abstract}
	
	\noindent \emph{Keywords}: Correlated equilibrium, optimin, Nash equilibrium, solution concept
	
	\section{Introduction}
	
	Correlated equilibrium extends Nash equilibrium by allowing players to condition their actions on private recommendations generated by a common random device \citep{nash1951,aumann1974}. \citet{aumann1974} showed that correlated equilibria can support payoffs outside the convex hull of Nash equilibrium payoffs and can, in fact, strictly Pareto dominate all Nash equilibria in a game. \citet{moulin1978} extended Aumann's notion to coarse correlated equilibrium, in which no player has an incentive to deviate from committing to the private recommendations before their realization.
	
	This paper asks a complementary question. Instead of requiring obedience after every recommendation, we evaluate a correlated distribution by the payoffs it guarantees when opponents may exploit the recommendation device in individually profitable ways. The starting point is the optimin criterion of \citet{ismail2025}, which evaluates a strategy profile by each player's worst payoff under unilateral profitable deviations by the other players and selects Pareto optimal profiles for this guarantee vector. We extend this criterion from independent mixed profiles to correlated distributions.
	
	The key difference from the mixed-profile setting is that, under a correlated distribution, a unilateral deviation need not be a fixed alternative action; it may depend on the player's privately observed recommendation. Given a distribution $P$, if player $j$ receives recommendation $a_j$, then $j$ may either obey or switch to an action that strictly improves $j$'s conditional expected payoff under the posterior induced by $P$. Player $i$'s correlated optimin guarantee is the worst ex ante payoff obtained when the other players choose admissible recommendation-contingent deviations. A correlated optimin is Pareto optimal with respect to this vector of guaranteed payoffs.
	
	We show that a correlated optimin exists in every finite game. In addition, for every correlated equilibrium, there exists a correlated optimin such that every player's guaranteed payoff is weakly higher than his or her correlated-equilibrium payoff. In two-player zero-sum games, correlated optimin coincides with correlated equilibrium and yields the maximin value. Outside zero-sum games, correlated optimin may strictly improve upon all correlated equilibria. The illustrative $2\times2$ game below has a correlated optimin point that strictly Pareto dominates the unique correlated equilibrium.
	
	\[
	\begin{array}{c|cc}
		& L & R\\
		\hline
		U & (0,1) & (1,0)\\
		D & (1,0) & (0,2)
	\end{array}.
	\]
	This game has a unique mixed Nash equilibrium, given by $\left[\left(\frac23,\frac13\right), \left(\frac12,\frac12\right)\right]$, with payoff vector $\left(\frac12,\frac23\right).$
	It also has a unique correlated equilibrium,
	\[
	P^{CE}=
	\begin{pmatrix}
		1/3&1/3\\
		1/6&1/6
	\end{pmatrix},
	\]
	whose payoff vector is $\left(\frac12,\frac23\right).$ However, the correlated distribution
	\[
	P=
	\begin{pmatrix}
		7/30&3/10\\
		7/30&7/30
	\end{pmatrix}
	\]
	is a correlated optimin point for which both the expected payoffs and the guaranteed payoffs coincide at $\left(\frac8{15},\frac7{10}\right)$,
	which strictly Pareto dominates the correlated equilibrium payoff.
	
	The intuition is simple. Under $P$, player $1$ has no strictly profitable unilateral deviation after either recommendation. By contrast, player $2$ has a strictly profitable recommendation-contingent deviation only when recommendation $L$ is received, in which case switching from $L$ to $R$ is profitable. Nevertheless, even after allowing player $2$'s profitable deviation, player $1$ still guarantees the payoff $8/15$. Thus, $P$ improves not only the actual payoffs relative to correlated equilibrium, but also the players' guaranteed payoffs under unilateral profitable deviations.
	
	\section{Setup}
	
	Let $\Gamma=(N,(A_i)_{i\in N},(u_i)_{i\in N})$ be a finite normal-form game.  The player set is $N=\{1,\dots,n\}$, $n\ge1$. Each $A_i$ is finite and nonempty, and $A=\prod_{i\in N}A_i$ is the finite set of pure action profiles.  Payoffs are real-valued functions $u_i:A\to\mathbb R$.
	Since $A$ is finite, all payoffs are bounded and all expectations below are finite.
	
	A pure action profile is written $a=(a_i,a_{-i})$, and $a_{-i}\in A_{-i}:=\prod_{j\ne i}A_j$.
	A \emph{correlated distribution} is a probability distribution $P\in\Delta(A)$. The interpretation is that a mediator draws $a\sim P$ and privately recommends $a_i$ to player $i$.
	
	For each recommendation $a_i\in A_i$, define the marginal recommendation probability
	\[
	P_i(a_i)=\sum_{a_{-i}\in A_{-i}}P(a_i,a_{-i}).
	\]
	If $P_i(a_i)=0$, the posterior is not defined. If $P_i(a_i)>0$, the posterior over the other players' recommendations is
	\[
	P(a_{-i}\given a_i)
	=
	\frac{P(a_i,a_{-i})}{P_i(a_i)}.
	\]
	
	A player who receives recommendation $a_i$ is assumed to know both the recommendation and the distribution $P$.
	
	Given a set $X$ and a vector-valued function $f:X\to\mathbb R^n$, a point $\bar x\in X$ is \emph{Pareto optimal} (undominated) if there is no $y\in X$ such that $f_i(y)\ge f_i(\bar x)$ for every $i\in N$, with strict inequality for at least one player. 
	
	\subsection{Correlated and coarse correlated equilibrium}
	
	\begin{definition}[Correlated equilibrium]
		A distribution $P\in\Delta(A)$ is a \emph{correlated equilibrium} if, for every player $i\in N$, every recommendation $a_i\in A_i$ with $P_i(a_i)>0$, and every deviation $b_i\in A_i$,
		\[
		\E_P[u_i(a_i,a_{-i})\given a_i]
		\ge
		\E_P[u_i(b_i,a_{-i})\given a_i].
		\]
	\end{definition}
	
	Multiplying by $P_i(a_i)>0$ gives the equivalent joint-probability inequality
	\[
	\sum_{a_{-i}\in A_{-i}}
	P(a_i,a_{-i})
	\bigl[u_i(a_i,a_{-i})-u_i(b_i,a_{-i})\bigr]
	\ge0.
	\]
	It is well-known that the set of correlated equilibria of $\Gamma$, denoted as $\CE(\Gamma)$, is a polytope.
	
	\begin{definition}[Coarse correlated equilibrium]
		A distribution $P\in\Delta(A)$ is a \emph{coarse correlated equilibrium} if, for every player $i\in N$ and every fixed action $b_i\in A_i$,
		\[
		\sum_{a\in A}P(a)u_i(a)
		\ge
		\sum_{a\in A}P(a)u_i(b_i,a_{-i}).
		\]
	\end{definition}
	
	The difference between the two concepts is the timing of deviations.  Correlated equilibrium compares obedience with deviations after the recommendation is observed, whereas coarse correlated equilibrium compares obedience with fixed commitments chosen before the recommendation is observed.
	
	A mixed Nash equilibrium $p=(p_i)_{i\in N}$ induces the product distribution $P(a)=\prod_i p_i(a_i)$.  If $p$ is Nash, the induced distribution is a correlated equilibrium.  Conditional on any positive-probability recommendation $a_i$, independence leaves player $i$'s posterior over $a_{-i}$ equal to $p_{-i}$, and the Nash best-response condition gives the correlated equilibrium inequalities.

	\section{Main definition}
	In the original optimin criterion, a mixed profile is evaluated by the worst-case payoff player $i$ can receive, if the other players either keep their original strategies or make unilateral deviations that are strictly profitable relative to the original profile. 
	
	We next extend this idea to correlated strategies. The main change is that a deviation is now allowed to depend on the recommendation privately observed by the deviating player.  Thus, admissibility is checked recommendation by recommendation, while the payoff of the evaluated player is computed after all other players' admissible rules are applied simultaneously.
	
	Fix $P\in\Delta(A)$.  If player $j$ receives recommendation $a_j$, they may either obey $a_j$ or switch to a pure action that is strictly profitable given the posterior induced by $P$.
	
	For player $j$, recommendation $a_j\in A_j$, and alternative action $b_j\in A_j$, define the unnormalized gain from switching to $b_j$ by
	\[
	\sum_{a_{-j}\in A_{-j}}
	P(a_j,a_{-j})
	\bigl[u_j(b_j,a_{-j})-u_j(a_j,a_{-j})\bigr].
	\]
	If $P_j(a_j)>0$, this expression is $P_j(a_j)$ times the conditional expected payoff gain from deviating after recommendation $a_j$.  Therefore, strict positivity is equivalent to strict conditional profitability.  If $P_j(a_j)=0$, the expression is zero for every $b_j$, so no strict deviation is added at that recommendation.
	
	Define
	\[
	B_j^P(a_j)
	=
	\{a_j\}
	\cup
	\left\{
	b_j\in A_j:
	\sum_{a_{-j}\in A_{-j}}
	P(a_j,a_{-j})
	\bigl[u_j(b_j,a_{-j})-u_j(a_j,a_{-j})\bigr]>0
	\right\}.
	\]
	Thus, $B_j^P(a_j)$ contains obedience plus all strictly profitable pure deviations after recommendation $a_j$.  At zero-probability recommendations the definition gives $B_j^P(a_j)=\{a_j\}$.
	
	\begin{remark}
		Since zero-probability recommendations never occur under $P$, they do not affect realized payoffs.  We therefore adopt the convention $B_j^P(a_j)=\{a_j\}$ whenever $P_j(a_j)=0$, so that only obedience is admissible off support.
	\end{remark}
	
	For player $i$, define the admissible pure recommendation-contingent profiles of the other players by
	\[
	B_{-i}^P
	:=
	\prod_{j\ne i}\prod_{a_j\in A_j}B_j^P(a_j).
	\]
	An element $\delta_{-i}\in B_{-i}^P$ is interpreted as a family of choices
	\[
	\delta_j(a_j)\in B_j^P(a_j),
	\qquad j\ne i,
	\quad a_j\in A_j.
	\]
	Equivalently, $\delta_j$ is a recommendation-contingent rule $A_j\to A_j$.  Given $a_{-i}=(a_j)_{j\ne i}$, write
	\[
	\delta_{-i}(a_{-i})=(\delta_j(a_j))_{j\ne i}.
	\]
	The obedient profile $\delta_j(a_j)=a_j$ for all $j\ne i$ and all $a_j$ is always admissible.  Since the action sets are finite, $B_{-i}^P$ is finite and nonempty.
	
	This construction respects the private-information structure of correlated equilibrium: player $j$'s action may depend on $j$'s own recommendation $a_j$, but not on recommendations observed only by other players.
	
	\begin{definition}[Correlated optimin performance]
		For $P\in\Delta(A)$, define
		\[
		\pi_i(P)
		=
		\min_{\delta_{-i}\in B_{-i}^P}
		\sum_{a\in A}P(a)
		u_i\bigl(a_i,\delta_{-i}(a_{-i})\bigr).
		\]
		A distribution $\bar P\in\Delta(A)$ is a \emph{correlated optimin} if it is Pareto optimal with respect to
		\[
		\pi(P)
		=
		(\pi_1(P),\dots,\pi_n(P)).
		\]
		The set of correlated optimins is denoted $\CO(\Gamma)$.
	\end{definition}
	
	Note that conditional probabilities enter in determining which deviations are admissible.  Once an admissible global deviation profile $\delta_{-i}$ has been fixed, the realized payoff to player $i$ is the random variable $u_i(a_i,\delta_{-i}(a_{-i}))$ under the original draw $a\sim P$.  Therefore the payoff is the ordinary ex ante expectation $\sum_{a\in A}P(a)u_i(a_i,\delta_{-i}(a_{-i}))$.
	
	Note also that the minimization is outside this expectation: a single global admissible deviation profile is chosen against the whole distribution. This is the closest correlated analogue of the original optimin criterion. The alternative `inside' formulation, discussed in the appendix, chooses a possibly different minimizing rule after each recommendation of player $i$ and is therefore more pessimistic. We introduce this version as an additional robustness check because, it is not reasonable to condition player $j$'s deviation on the information provided to player $i\neq j$.
	
	\section{Main results}
	
	\begin{proposition}[Performance at correlated equilibria]\label{prop:collapse}
		If $P\in\CE(\Gamma)$, then for every player $i$,
		\[
		\pi_i(P)=\sum_{a\in A}P(a)u_i(a).
		\]
	\end{proposition}
	
	\begin{proof}
		Fix a correlated equilibrium $P$.  For every player $j$, every recommendation $a_j$ with $P_j(a_j)>0$, and every action $b_j\in A_j$, correlated equilibrium gives
		\[
		\E_P[u_j(a_j,a_{-j})\given a_j]
		\ge
		\E_P[u_j(b_j,a_{-j})\given a_j].
		\]
		Multiplying by $P_j(a_j)>0$ gives
		\[
		\sum_{a_{-j}}P(a_j,a_{-j})
		\bigl[u_j(b_j,a_{-j})-u_j(a_j,a_{-j})\bigr]
		\le0.
		\]
		If $P_j(a_j)=0$, the same joint-gain expression is exactly zero for every $b_j$.  Hence no action $b_j\ne a_j$ satisfies the strict inequality defining $B_j^P(a_j)$, and therefore
		\[
		B_j^P(a_j)=\{a_j\}
		\quad\text{for every }j\text{ and every }a_j.
		\]
		It follows that every admissible recommendation-contingent profile in $B_{-i}^P$ agrees with obedience on every recommendation.  Substituting $\delta_{-i}(a_{-i})=a_{-i}$ into the definition of $\pi_i$ yields
		\[
		\pi_i(P)
		=
		\sum_{a\in A}P(a)u_i(a_i,a_{-i})
		=
		\sum_{a\in A}P(a)u_i(a),
		\]
		as desired.
	\end{proof}
	
	At a correlated equilibrium, no recommendation creates a strictly profitable unilateral deviation.  Since the correlated optimin admissible set adds only strictly profitable deviations, the admissible set collapses to obedience.  The worst case is then the same as the realized expected payoff.
	
	\begin{proposition}[Upper semicontinuity]\label{prop:usc}
		For every player $i$, the function
		\[
		P\mapsto \pi_i(P)
		\]
		is upper semicontinuous on $\Delta(A)$.
	\end{proposition}
	
	\begin{proof}
		Let $P^m\to P$ in $\Delta(A)$.  We must show
		\[
		\limsup_{m\to\infty}\pi_i(P^m)
		\le
		\pi_i(P).
		\]
		Because $B_{-i}^P$ is finite and nonempty, the minimum defining $\pi_i(P)$ is attained.  Choose $\delta_{-i}\in B_{-i}^P$
		such that
		\[
		\pi_i(P)
		=
		\sum_{a\in A}P(a)u_i(a_i,\delta_{-i}(a_{-i})).
		\]
		
		The key step is to prove that this same rule profile remains admissible for all sufficiently large $m$.  Fix an opponent $j\ne i$ and a recommendation $a_j\in A_j$.
		
		If $\delta_j(a_j)=a_j$, this component is obedient and hence belongs to $B_j^{P^m}(a_j)$ for every distribution $P^m$.
		
		If instead $\delta_j(a_j)=b_j\ne a_j$, then admissibility under $P$ means that the strict joint-gain inequality holds:
		\[
		\sum_{a_{-j}\in A_{-j}}
		P(a_j,a_{-j})
		\bigl[u_j(b_j,a_{-j})-u_j(a_j,a_{-j})\bigr]>0.
		\]
		In particular, this case cannot occur at a zero-probability recommendation under $P$, because then the left-hand side would be zero.  The left-hand side is a finite linear function of the coordinates of $P$.  Therefore it is continuous in $P$.  Since it is strictly positive at $P$, it remains strictly positive at $P^m$ for all sufficiently large $m$:
		\[
		\sum_{a_{-j}\in A_{-j}}
		P^m(a_j,a_{-j})
		\bigl[u_j(b_j,a_{-j})-u_j(a_j,a_{-j})\bigr]>0.
		\]
		Thus, $b_j\in B_j^{P^m}(a_j)$ for all sufficiently large $m$.
		
		There are only finitely many pairs $(j,a_j)$ for $j\ne i$.  Taking the maximum of the finitely many thresholds obtained above, we conclude that $\delta_{-i}\in B_{-i}^{P^m}$ for all sufficiently large $m$.
		
		Hence, for all sufficiently large $m$,
		\[
		\pi_i(P^m)
		\le
		\sum_{a\in A}P^m(a)u_i(a_i,\delta_{-i}(a_{-i})).
		\]
		The right-hand side is a finite linear function of $P^m$, so it converges to
		\[
		\sum_{a\in A}P(a)u_i(a_i,\delta_{-i}(a_{-i}))
		=
		\pi_i(P).
		\]
		Taking limit superior gives
		\[
		\limsup_m\pi_i(P^m)
		\le
		\pi_i(P),
		\]
		as required.
	\end{proof}

	A lower jump may occur if new profitable deviations become admissible near $P$, since the minimum is then taken over a larger set of deviation rules.  Upper semicontinuity permits such downward jumps but rules out upward jumps.  Every nonobedient component of a minimizing admissible rule at $P$ is supported by a strict profitability inequality, and strict inequalities persist under sufficiently small perturbations.  Consequently, the same minimizing rule remains admissible at nearby distributions, yielding a uniform upper bound on nearby values of $\pi_i$.
	
	\begin{proposition}[Existence]\label{prop:existence}
		The set $\CO(\Gamma)$ is nonempty.
	\end{proposition}
	
	\begin{proof}
		Consider the scalar function
		\[
		W(P)=\sum_{i\in N}\pi_i(P).
		\]
		Each $\pi_i$ is upper semicontinuous by Proposition~\ref{prop:usc}, so $W$ is upper semicontinuous. Since $\Delta(A)$ is compact, $W$ attains a maximum on $\Delta(A)$. Let $\bar P$ be a maximizer.
		
		If $\bar P$ were Pareto dominated with respect to $\pi$, then there would exist $Q\in\Delta(A)$ such that $\pi_i(Q)\ge\pi_i(\bar P)$ for every $i$, with strict inequality for at least one player. Summing over players would give $W(Q)>W(\bar P)$, contradicting the maximality of $\bar P$. Therefore $\bar P\in\CO(\Gamma)$.
	\end{proof}
	
	The proof uses a standard compactness argument.  Upper semicontinuity ensures that the aggregate performance function $W$ attains a maximum on the compact simplex $\Delta(A)$.  The function $W$ is used only as a selection device: any maximizer of $W$ must be Pareto optimal with respect to $\pi$, because a distribution that weakly improved every player's guaranteed payoff and strictly improved at least one would necessarily yield a strictly larger value of $W$.
	
	\begin{corollary}[Correlated equilibrium domination]\label{cor:ce-improvement}
		For every correlated equilibrium $P^{\CE}\in\CE(\Gamma)$, there exists a correlated optimin $\bar P\in\CO(\Gamma)$ such that
		\[
		\pi_i(\bar P)
		\ge
		\sum_{a\in A}P^{\CE}(a)u_i(a)
		\quad\text{for every }i\in N.
		\]
		Moreover,
		\[
		\sum_{a\in A}\bar P(a)u_i(a)
		\ge
		\pi_i(\bar P)
		\quad\text{for every }i,
		\]
		so the same distribution also weakly improves every player's ordinary expected payoff relative to $P^{\CE}$.
	\end{corollary}
	
	\begin{proof}
		By Proposition~\ref{prop:collapse},
		\[
		\pi_i(P^{\CE})
		=
		\sum_{a\in A}P^{\CE}(a)u_i(a)
		\]
		for every player $i$.
		
		Define
		\[
		S(P^{\CE})
		=
		\left\{
		Q\in\Delta(A):
		\pi_i(Q)\ge \pi_i(P^{\CE})
		\text{ for every }i
		\right\}.
		\]
		This set is nonempty because $P^{\CE}\in S(P^{\CE})$. By upper semicontinuity of each $\pi_i$, it is closed. Since $\Delta(A)$ is compact, $S(P^{\CE})$ is compact.
		
		Let $W(P)=\sum_i\pi_i(P)$. Since $W$ is upper semicontinuous, it attains a maximum on $S(P^{\CE})$. Let $\bar P$ be a maximizer. Then $\pi_i(\bar P)\ge\pi_i(P^{\CE})$ for every $i$.
		
		It remains to show that $\bar P$ is globally Pareto optimal. Suppose not. Then some $Q\in\Delta(A)$ weakly improves every component of $\pi(\bar P)$ and strictly improves at least one. Since $\bar P\in S(P^{\CE})$, this implies $Q\in S(P^{\CE})$. But then $W(Q)>W(\bar P)$, contradicting the maximality of $\bar P$ on $S(P^{\CE})$. Hence $\bar P\in\CO(\Gamma)$.
		
		Finally, obedience by all opponents is always admissible, so
		\[
		\pi_i(\bar P)
		\le
		\sum_{a\in A}\bar P(a)u_i(a).
		\]
		Combining the inequalities gives the result.
	\end{proof}

	At a correlated equilibrium, Proposition~\ref{prop:collapse} implies that guaranteed payoffs coincide with ordinary expected payoffs.  The set of distributions whose guaranteed-performance vectors weakly improve upon the correlated-equilibrium payoff vector is compact by upper semicontinuity.  Maximizing the aggregate guaranteed payoff over this set therefore selects a globally Pareto-optimal distribution.  Since obedience is always admissible in the definition of $\pi_i$, each player's ordinary expected payoff at the selected distribution is at least as large as their guaranteed payoff.  Thus the same distribution weakly improves the correlated-equilibrium payoff both in guaranteed-performance terms and in ordinary expected-payoff terms.
	
	\subsection{Two-player zero-sum games}
	
	\begin{proposition}[Two-player zero-sum games]\label{prop:zero-sum}
		Let $\Gamma$ be a finite two-player zero-sum game.  Player $1$'s payoff is $u:A_1\times A_2\to\mathbb R$, and player $2$'s payoff is $-u$.  Let
		\[
		v=
		\max_{p_1\in\Delta(A_1)}\min_{a_2\in A_2}u(p_1,a_2)
		=
		\min_{p_2\in\Delta(A_2)}\max_{a_1\in A_1}u(a_1,p_2)
		\]
		be the value.  Then every correlated equilibrium has payoff vector $(v,-v)$, and
		\[
		\CO(\Gamma)=\CE(\Gamma).
		\]
	\end{proposition}
	
	\begin{proof}
		Let $P\in\Delta(A_1\times A_2)$ be arbitrary, and write
		\[
		e(P)=\sum_{a_1,a_2}P(a_1,a_2)u(a_1,a_2)
		\]
		for player $1$'s expected payoff under obedience.
		
		First compute the correlated optimin performance in zero-sum form.  For player $1$, the only opponent is player $2$.  Since player $2$'s payoff is $-u$, a deviation by player $2$ is strictly profitable exactly when it lowers player $1$'s payoff.  Therefore, after a recommendation $a_2$, the worst admissible action for player $1$ is represented by the minimum of player $1$'s joint payoff contribution:
		\[
		\pi_1(P)
		=
		\sum_{a_2\in A_2}
		\min_{b_2\in A_2}
		\sum_{a_1\in A_1}P(a_1,a_2)u(a_1,b_2).
		\]
		This formula is valid even though the definition admits only strictly profitable deviations plus obedience.  If some action gives a lower payoff than obedience, it is strictly profitable for player $2$ and is admissible; if no action gives a lower payoff, obedience itself attains the minimum.
		
		Similarly, player $2$'s worst case comes from player $1$ choosing actions that maximize $u$ recommendation by recommendation.  Thus
		\[
		\pi_2(P)
		=
		-
		\sum_{a_1\in A_1}
		\max_{b_1\in A_1}
		\sum_{a_2\in A_2}P(a_1,a_2)u(b_1,a_2).
		\]
		
		We now bound the two components by the value.  For player $1$,
		\[
		\begin{aligned}
			\pi_1(P)
			&=
			\sum_{a_2}\min_{b_2}\sum_{a_1}P(a_1,a_2)u(a_1,b_2)\\
			&\le
			\min_{b_2}\sum_{a_2}\sum_{a_1}P(a_1,a_2)u(a_1,b_2)\\
			&=
			\min_{b_2}\sum_{a_1}P_1(a_1)u(a_1,b_2)\\
			&\le
			\max_{p_1\in\Delta(A_1)}\min_{b_2\in A_2}u(p_1,b_2)
			=v.
		\end{aligned}
		\]
		For player $2$, define
		\[
		B(P)=-\pi_2(P)
		=
		\sum_{a_1}\max_{b_1}\sum_{a_2}P(a_1,a_2)u(b_1,a_2).
		\]
		Then
		\[
		\begin{aligned}
			B(P)
			&\ge
			\max_{b_1}\sum_{a_1}\sum_{a_2}P(a_1,a_2)u(b_1,a_2)\\
			&=
			\max_{b_1}\sum_{a_2}P_2(a_2)u(b_1,a_2)\\
			&\ge
			\min_{p_2\in\Delta(A_2)}\max_{b_1\in A_1}u(b_1,p_2)
			=v.
		\end{aligned}
		\]
		Thus, for every correlated distribution $P$,
		\[
		\pi_1(P)\le v,
		\qquad
		\pi_2(P)\le -v.
		\]
		
		Now let $P^{\CE}$ be a correlated equilibrium.  By Proposition~\ref{prop:collapse},
		\[
		\pi(P^{\CE})=(e(P^{\CE}),-e(P^{\CE})).
		\]
		The bounds just proved imply
		\[
		e(P^{\CE})\le v
		\quad\text{and}\quad
		-e(P^{\CE})\le -v,
		\]
		which is equivalent to $e(P^{\CE})\ge v$.  Hence $e(P^{\CE})=v$, so every correlated equilibrium has payoff vector $(v,-v)$.
		
		We next prove the set equality.  First, every correlated equilibrium is correlated optimin.  Indeed, every correlated equilibrium has performance $(v,-v)$, and no distribution can have first component greater than $v$ or second component greater than $-v$.  Hence no distribution can Pareto dominate a correlated equilibrium.
		
		Conversely, let $P\in\CO(\Gamma)$.  Since finite games have correlated equilibria, there exists $P^{\CE}$ with performance $(v,-v)$.  If either $\pi_1(P)<v$ or $\pi_2(P)<-v$, then $P^{\CE}$ would weakly improve both components and strictly improve at least one, contradicting Pareto undominatedness of $P$.  Therefore
		\[
		\pi(P)=(v,-v).
		\]
		
		It remains to show that this equality forces $P$ to satisfy the correlated equilibrium obedience constraints.  Define
		\[
		A(P)=\pi_1(P),
		\qquad
		B(P)=-\pi_2(P).
		\]
		For each recommendation $a_2$ of player $2$,
		\[
		\min_{b_2}\sum_{a_1}P(a_1,a_2)u(a_1,b_2)
		\le
		\sum_{a_1}P(a_1,a_2)u(a_1,a_2),
		\]
		so summing over $a_2$ gives $A(P)\le e(P)$.  For each recommendation $a_1$ of player $1$,
		\[
		\max_{b_1}\sum_{a_2}P(a_1,a_2)u(b_1,a_2)
		\ge
		\sum_{a_2}P(a_1,a_2)u(a_1,a_2),
		\]
		so summing over $a_1$ gives $B(P)\ge e(P)$.  Since $A(P)=B(P)=v$, we obtain
		\[
		A(P)=e(P)=B(P)=v.
		\]
		The inequalities above are termwise before summing.  Equality of the sums therefore forces equality at every recommendation.  Hence obedience attains the recommendation-level minimum for player $2$ and the recommendation-level maximum for player $1$.  In payoff terms, player $2$ has no strictly profitable deviation after any recommendation, and player $1$ has no strictly profitable deviation after any recommendation.  Thus, $P$ is a correlated equilibrium.
	\end{proof}
	
	In a zero-sum game, a profitable deviation by one player is exactly a deviation that lowers the other player's payoff.  Correlated optimin therefore reduces to asking how much each player can guarantee recommendation by recommendation against profitable deviations by the opponent.  These guarantees cannot improve upon the value of the zero-sum game, and correlated equilibria already attain that value.  Consequently, correlated optimin and correlated equilibrium coincide in finite two-player zero-sum games.
	
	\section{Correlated optimin vs. correlated equilibrium: strict Pareto domination}
	
	Consider the two-player game
	\[
	\begin{array}{c|cc}
		& L & R\\
		\hline
		U & (0,1) & (1,0)\\
		D & (1,0) & (0,2)
	\end{array}
	\]
	A correlated distribution is denoted by
	\[
	P=
	\begin{pmatrix}
		x&y\\
		z&w
	\end{pmatrix},
	\]
	where $x=P(U,L)$, $y=P(U,R)$, $z=P(D,L)$, and $w=P(D,R)$. The probabilities satisfy $x,y,z,w\ge0$ and $x+y+z+w=1$.
	
	The correlated equilibrium inequalities are as follows.  If player $1$ is recommended $U$, obedience gives weighted payoff $y$, while switching to $D$ gives weighted payoff $x$.  Hence $y\ge x$. If player $1$ is recommended $D$, obedience gives $z$, while switching to $U$ gives $w$. Hence $z\ge w$. If player $2$ is recommended $L$, obedience gives $x$, while switching to $R$ gives $2z$.  Hence $x\ge 2z$.  If player $2$ is recommended $R$, obedience gives $2w$, while switching to $L$ gives $y$.  Hence $2w\ge y$.
	Therefore the CE system is
	\[
	y\ge x,
	\qquad
	z\ge w,
	\qquad
	x\ge2z,
	\qquad
	2w\ge y.
	\]
	Combining them gives
	\[
	y\ge x\ge2z\ge2w\ge y.
	\]
	All inequalities must bind.  Thus,  $y=x$, $x=2z$, and $z=w$.  Since $x+y+z+w=1$,
	\[
	2z+2z+z+z=6z=1.
	\]
	So $z=w=\frac16$ and $x=y=\frac13$. The unique correlated equilibrium is therefore
	\[
	P^{\CE}
	=
	\begin{pmatrix}
		1/3&1/3\\
		1/6&1/6
	\end{pmatrix}.
	\]
	Its payoff vector is
	\[
	\E_{P^{\CE}}[u_1]
	=0\cdot\frac13+1\cdot\frac13+1\cdot\frac16+0\cdot\frac16
	=\frac12,
	\]
	and
	\[
	\E_{P^{\CE}}[u_2]
	=1\cdot\frac13+0\cdot\frac13+0\cdot\frac16+2\cdot\frac16
	=\frac23.
	\]
	Thus,  $\E_{P^{\CE}}[u]=\left(\frac12,\frac23\right)$.
	
	Now consider
	\[
	Q=
	\begin{pmatrix}
		7/30&3/10\\
		7/30&7/30
	\end{pmatrix}
	\]
	Equivalently, $x=\frac7{30},
	y=\frac9{30},
	z=\frac7{30}$, and
	$w=\frac7{30}$. We first compute the correlated optimin performance of $Q$.
	
	For player $1$, the only strictly profitable deviation of player $2$ is from $L$ to $R$.  Indeed, conditional on recommendation $L$, player $2$ obtains payoff $1/2$ from obedience and $1$ from switching to $R$.  Conditional on recommendation $R$, obedience gives $7/8$ while switching to $L$ gives $9/16$, so from $R$ to $L$ is not strictly profitable.
	
	If player $2$ obeys, player $1$'s payoff is
	\[
	\E_Q[u_1]=y+z=\frac9{30}+\frac7{30}=\frac8{15}.
	\]
	If player $2$ switches from $L$ to $R$ and obeys $R$, then $(U,L)$ and $(D,L)$ become $(U,R)$ and $(D,R)$.  Since $x=z=7/30$, player $1$'s payoff remains
	\[
	x+y=\frac7{30}+\frac9{30}=\frac8{15}.
	\]
	Hence $\pi_1(Q)=\frac8{15}$.
	
	For player $2$, player $1$ has no strictly profitable deviation after either recommendation.  Conditional on $U$, obedience gives player $1$ payoff $9/16$ while switching to $D$ gives $7/16$.  Conditional on $D$, both obedience and switching to $U$ give $1/2$.  Thus, obedience is player $1$'s only admissible rule, and therefore
	\[
	\pi_2(Q)=\E_Q[u_2]=x+2w
	=\frac7{30}+\frac{14}{30}
	=\frac7{10}.
	\]
	
	Therefore
	\[
	\pi(Q)=\left(\frac8{15},\frac7{10}\right).
	\]
	Since $\frac8{15}>\frac12$ and $\frac7{10}>\frac23$, the guaranteed-performance vector of $Q$ strictly Pareto dominates the unique correlated equilibrium payoff.
	
	We now show that $Q$ is itself a correlated optimin.  Let
	\[
	P=
	\begin{pmatrix}
		x&y\\
		z&w
	\end{pmatrix}
	\]
	be an arbitrary correlated distribution.  Since obedience is always admissible,
	\[
	\pi_1(P)\le y+z,
	\qquad
	\pi_2(P)\le x+2w.
	\]
	Hence
	\[
	3\pi_1(P)+2\pi_2(P)
	\le
	3(y+z)+2(x+2w).
	\]
	Using $x+y+z+w=1$,
	\[
	3(y+z)+2(x+2w)
	=
	3+(w-x).
	\]
	
	If $w\le x$, this immediately gives
	\[
	3\pi_1(P)+2\pi_2(P)\le3.
	\]
	
	Suppose instead that $w>x$.  If $z\ge w$, then player $2$'s deviation from $L$ to $R$ is strictly profitable, and player $1$'s guaranteed payoff falls by at least $z-x\ge w-x$.  If $z<w$, then player $1$'s deviation from $D$ to $U$ is strictly profitable, and player $2$'s guaranteed payoff falls by at least $2w-z>w-x$.  In either case, the deviation losses offset the excess term $w-x$, so again
	\[
	3\pi_1(P)+2\pi_2(P)\le3.
	\]
	
	Thus, every correlated distribution satisfies
	\[
	3\pi_1(P)+2\pi_2(P)\le3.
	\]
	At $Q$,
	\[
	3\pi_1(Q)+2\pi_2(Q)
	=
	3\cdot\frac8{15}
	+
	2\cdot\frac7{10}
	=
	3.
	\]
	Since both weights are strictly positive, no distribution can weakly improve both components of $\pi(Q)$ while strictly improving one.  Therefore $Q$ is Pareto-undominated with respect to $\pi$, and hence
	\[
	Q\in\CO(\Gamma).
	\]
	
	Consequently,
	\[
	\pi(Q)
	=
	\left(\frac8{15},\frac7{10}\right)
	\gg
	\left(\frac12,\frac23\right)
	=
	\E_{P^{\CE}}[u].
	\]
	
	Thus, the game admits a correlated optimin that strictly Pareto dominates the unique correlated equilibrium in both payoffs and guaranteed-performance terms.

	\appendix
	
	\section{Appendix: Variants}\label{app:variants}
	
	This appendix illustrates variants that are useful for comparison but are not used as the main concept. The default version is preferred because it is the most natural one considering the original optimin criterion and correlated equilibrium.
	
	\subsection{Mixed-deviation correlated optimin}
	
	The mixed-deviation version lets a player switch after a recommendation to a mixed action.  For $q_j\in\Delta(A_j)$, write
	\[
	u_j(q_j,a_{-j})=
	\sum_{b_j\in A_j}q_j(b_j)u_j(b_j,a_{-j}).
	\]
	Let $e_{a_j}$ denote the degenerate mixed action on $a_j$.  Define
	\[
	M_j^P(a_j)
	=
	\{e_{a_j}\}
	\cup
	\left\{
	q_j\in\Delta(A_j):
	\sum_{a_{-j}\in A_{-j}}
	P(a_j,a_{-j})
	\bigl[u_j(q_j,a_{-j})-u_j(a_j,a_{-j})\bigr]>0
	\right\}.
	\]
	At zero-probability recommendations this reduces to $\{e_{a_j}\}$.  For player $i$, set
	\[
	M_{-i}^P=
	\prod_{j\ne i}\prod_{a_j\in A_j}M_j^P(a_j).
	\]
	An element $\sigma_{-i}\in M_{-i}^P$ is interpreted as a family of mixed recommendation-contingent rules $\sigma_j(a_j)\in M_j^P(a_j)$.
	
	We assume independent randomization across deviating players conditional on their private recommendations.  Thus, given $a_{-i}$, the probability of realized action profile $b_{-i}$ is
	\[
	\prod_{j\ne i}\sigma_j(a_j)(b_j).
	\]
	
	Define
	\[
	U_i(a_i,\sigma_{-i}(a_{-i}))
	=
	\sum_{b_{-i}\in A_{-i}}
	\left(\prod_{j\ne i}\sigma_j(a_j)(b_j)\right)
	u_i(a_i,b_{-i}),
	\]
	and
	\[
	\pi_i^{m}(P)
	=
	\inf_{\sigma_{-i}\in M_{-i}^P}
	\sum_{a\in A}P(a)U_i(a_i,\sigma_{-i}(a_{-i})).
	\]
	The infimum need not be attained because strict-gain sets in mixed-action simplices need not be closed.  The value is nevertheless finite: obedience is always admissible and payoffs are bounded.  A distribution is a \emph{mixed-deviation  correlated optimin} if it is Pareto optimal with respect to $\pi^{m}$; the set of such distributions is denoted $\CO^{m}(\Gamma)$.
	
	The next propositions record the two basic facts for this variant.  They are included here because the mixed admissible sets need not be closed, so one must use an infimum and an $\varepsilon$-optimal rule rather than a minimizing rule.
	
	\begin{proposition}[Mixed-deviation collapse at correlated equilibria]
		If $P\in\CE(\Gamma)$, then for every player $i$,
		\[
		\pi_i^{m}(P)=\sum_{a\in A}P(a)u_i(a).
		\]
	\end{proposition}
	
	\begin{proof}
		At a correlated equilibrium, every pure deviation after every positive-probability recommendation has weakly nonpositive gain.  At a zero-probability recommendation the joint gain is zero.  Hence for every $j,a_j,b_j$,
		\[
		\sum_{a_{-j}}P(a_j,a_{-j})
		\bigl[u_j(b_j,a_{-j})-u_j(a_j,a_{-j})\bigr]\le0.
		\]
		For a mixed action $q_j\in\Delta(A_j)$, the mixed gain is the convex combination of these pure gains:
		\[
		\sum_{a_{-j}}P(a_j,a_{-j})
		\bigl[u_j(q_j,a_{-j})-u_j(a_j,a_{-j})\bigr]
		=
		\sum_{b_j\in A_j}q_j(b_j)
		\sum_{a_{-j}}P(a_j,a_{-j})
		\bigl[u_j(b_j,a_{-j})-u_j(a_j,a_{-j})\bigr].
		\]
		A convex combination of weakly nonpositive numbers is weakly nonpositive.  Thus, no nonobedient mixed action is strictly profitable, and $M_j^P(a_j)=\{e_{a_j}\}$ for every $j,a_j$.  The only admissible mixed rule profile is obedience, so the performance equals ordinary expected utility.
	\end{proof}
	
	\begin{proposition}[Upper semicontinuity for mixed deviations]
		For every player $i$, the function $P\mapsto\pi_i^{m}(P)$ is upper semicontinuous on $\Delta(A)$.
	\end{proposition}
	
	\begin{proof}
		Let $P^r\to P$.  We prove
		\[
		\limsup_{r\to\infty}\pi_i^{m}(P^r)
		\le
		\pi_i^{m}(P).
		\]
		Because the inner infimum need not be attained, fix $\varepsilon>0$ and choose an admissible mixed rule profile $\sigma_{-i}\in M_{-i}^P$ such that
		\[
		\sum_{a\in A}P(a)U_i(a_i,\sigma_{-i}(a_{-i}))
		\le
		\pi_i^{m}(P)+\varepsilon.
		\]
		We show that this same rule profile is admissible at all nearby $P^r$.
		
		Fix $j\ne i$ and $a_j\in A_j$.  If $\sigma_j(a_j)=e_{a_j}$, then this component is admissible under every distribution.  Otherwise, admissibility at $P$ implies the strict mixed-gain inequality
		\[
		\sum_{a_{-j}}P(a_j,a_{-j})
		\bigl[u_j(\sigma_j(a_j),a_{-j})-u_j(a_j,a_{-j})\bigr]>0.
		\]
		For fixed $\sigma_j(a_j)$, the left-hand side is linear and therefore continuous in $P$.  Hence the same strict inequality holds with $P^r$ in place of $P$ for all sufficiently large $r$.  Since there are finitely many pairs $(j,a_j)$, the whole rule profile is admissible under $P^r$ for all sufficiently large $r$.
		
		Therefore, for all sufficiently large $r$,
		\[
		\pi_i^{m}(P^r)
		\le
		\sum_{a\in A}P^r(a)U_i(a_i,\sigma_{-i}(a_{-i})).
		\]
		Taking the limit superior gives
		\[
		\limsup_r\pi_i^{m}(P^r)
		\le
		\sum_{a\in A}P(a)U_i(a_i,\sigma_{-i}(a_{-i}))
		\le
		\pi_i^{m}(P)+\varepsilon.
		\]
		Since $\varepsilon>0$ was arbitrary, upper semicontinuity follows.
	\end{proof}
	
	\begin{proposition}[Existence for mixed deviations]
		Every finite game has at least one mixed-deviation correlated optimin.
	\end{proposition}
	
	\begin{proof}
		The admissible mixed rule set is nonempty because obedience is always admissible, and all payoff functions are bounded because the game is finite.  Hence each $\pi_i^{m}$ is finite-valued.  By the previous proposition, each component is upper semicontinuous.  Therefore
		\[
		W^m(P)=\sum_{i\in N}\pi_i^{m}(P)
		\]
		is upper semicontinuous on compact $\Delta(A)$ and attains a maximum.  As in the main proof, the sum is only a selection device.  A maximizer cannot be Pareto dominated with respect to $\pi^{m}$; otherwise the dominating distribution would give a strictly larger value of $W^m$.  Thus, a mixed-deviation correlated optimin exists.
	\end{proof}

	\begin{proposition}\label{prop:2x2-mixed-pure}
		Let each of two players have exactly two pure actions.  Then, for every $P\in\Delta(A)$ and every player $i$,
		\[
		\pi_i^{m}(P)=\pi_i(P).
		\]
		Consequently, the mixed-deviation and pure-deviation correlated optimin sets coincide.
	\end{proposition}

	\begin{proof}
		Fix player $i$ and let $j\ne i$ be the unique opponent.  After any recommendation $a_j$, let $b_j$ be the only other pure action.  Every mixed action has the form $q_j=(1-r)e_{a_j}+r e_{b_j}$, $r\in[0,1]$.
		
		The strict-profitability expression for this mixed action is
		\[
		\sum_{a_{-j}}P(a_j,a_{-j})
		\bigl[u_j(q_j,a_{-j})-u_j(a_j,a_{-j})\bigr]
		=
		r
		\sum_{a_{-j}}P(a_j,a_{-j})
		\bigl[u_j(b_j,a_{-j})-u_j(a_j,a_{-j})\bigr].
		\]
		Hence a nonobedient mixed deviation is strictly profitable exactly when the pure switch from $a_j$ to $b_j$ is strictly profitable.
		
		For fixed $a_j$, player $i$'s payoff is affine in $r$.  Therefore the infimum over admissible mixed deviations is attained, or approached, at one of the endpoints $r=0$ or $r=1$, corresponding to obedience and the pure switch.  Thus, recommendation by recommendation, mixed deviations generate the same worst-case payoff as pure deviations.
		
		Since there is only one opponent, the objective separates across that opponent's recommendations.  Hence $\pi_i^m(P)=\pi_i(P)$ for every $P$ and every player $i$.  The equality of the Pareto optimal sets follows immediately.
	\end{proof}
	
	\subsection{Inside and statewise-selector variants}
	
	The main definition uses an outside minimum: one global deviation profile is chosen against the whole distribution.  The inside variant evaluates each recommendation of player $i$ separately:
	\[
	\pi_i^{\delta,\inn}(P)
	=
	\sum_{a_i\in A_i}
	\min_{\delta_{-i}\in B_{-i}^P}
	\sum_{a_{-i}\in A_{-i}}
	P(a_i,a_{-i})u_i(a_i,\delta_{-i}(a_{-i})).
	\]
	The joint-probability form avoids assigning posteriors to zero-probability recommendations.  The inside criterion is more pessimistic because different recommendations of player $i$ may be evaluated using different minimizing rules.  Consequently the inside value can combine conditional worst cases that do not arise from one global profile of deviation behavior.
	
	A statewise selector for player $i$ is a function
	\[
	\alpha_{-i}:A_{-i}\to A_{-i},
	\qquad
	\alpha_{-i}(a_{-i})=(\alpha_j(a_{-i}))_{j\ne i},
	\]
	with admissibility condition
	\[
	\alpha_j(a_{-i})\in B_j^P(a_j)
	\quad\text{for every }j\ne i\text{ and every }a_{-i}\in A_{-i}.
	\]
	Let $\mathcal A_{-i}(P)$ be the finite nonempty set of such selectors.  The outside and inside selector performances are
	\[
	\pi_i^{A}(P)
	=
	\min_{\alpha_{-i}\in\mathcal A_{-i}(P)}
	\sum_{a\in A}P(a)u_i(a_i,\alpha_{-i}(a_{-i}))
	\]
	and
	\[
	\pi_i^{A,\inn}(P)
	=
	\sum_{a_i\in A_i}
	\min_{\alpha_{-i}\in\mathcal A_{-i}(P)}
	\sum_{a_{-i}\in A_{-i}}P(a_i,a_{-i})u_i(a_i,\alpha_{-i}(a_{-i})).
	\]
	These selector variants are formal robustness benchmarks.  They are less behaviorally conservative because $\alpha_j(a_{-i})$ may depend on the full vector $a_{-i}$ even though player $j$ privately observes only $a_j$.
	
	\subsection{Coarse-correlated optimin}
	
	A coarse-correlated variant checks profitability ex ante rather than recommendation by recommendation.  For $b_j\in A_j$, define the ex ante pure gain
	\[
	H_j(P;b_j)=
	\sum_{a\in A}P(a)
	\bigl[u_j(b_j,a_{-j})-u_j(a_j,a_{-j})\bigr].
	\]
	Let $\iota_j(a_j)=a_j$ be obedience and let $\kappa_j^{b_j}(a_j)=b_j$ be the constant rule.  Define
	\[
	\mathcal C_j^p(P)=
	\{\iota_j\}
	\cup
	\{\kappa_j^{b_j}:b_j\in A_j,
	H_j(P;b_j)>0\},
	\qquad
	\mathcal C_{-i}^p(P)=\prod_{j\ne i}\mathcal C_j^p(P).
	\]
	The pure coarse-correlated outside performance is
	\[
	\pi_i^{\mathrm{cc},p}(P)
	=
	\min_{\gamma_{-i}\in\mathcal C_{-i}^p(P)}
	\sum_{a\in A}P(a)u_i(a_i,\gamma_{-i}(a_{-i})).
	\]
	If $P$ is a coarse correlated equilibrium, then $H_j(P;b_j)\le0$ for every $j,b_j$, so only obedience is admissible and this performance equals ordinary expected utility.
	
	The pure coarse-correlated performance is upper semicontinuous by the same finite strict-inequality argument as in the main proof.  To see this explicitly, let $P^m\to P$ and choose a minimizing rule profile $\gamma_{-i}\in\mathcal C_{-i}^p(P)$.  If a component $\gamma_j$ is obedience, it remains admissible under every $P^m$.  If $\gamma_j$ is a constant rule $\kappa_j^{b_j}$, then admissibility at $P$ means $H_j(P;b_j)>0$.  Since $H_j(\cdot;b_j)$ is linear in $P$, the same strict inequality holds for all sufficiently large $m$.  Hence the whole minimizing rule remains admissible near $P$, and the same limsup argument proves upper semicontinuity.  Compactness of $\Delta(A)$ then gives existence of a pure coarse-correlated  optimin distribution by maximizing the sum of the performance components. The coarse correlated optimin can be extended to mixed deviations analogously.

\end{document}